\documentclass[11pt]{article}
  
\usepackage{cite}
\usepackage{graphicx}
\usepackage{amsmath}
\usepackage{amsthm}
\usepackage{yhmath} 
\usepackage{hyperref}

\usepackage{amsfonts} 
\usepackage{amssymb}
\usepackage{fullpage}
\usepackage{color}
\usepackage{latexsym}

\usepackage{ccicons}
\usepackage{cclicenses}
\usepackage{algorithm}
\usepackage[noend]{algorithmic}
\usepackage{multirow} 
\usepackage{amsmath}
\usepackage{yhmath}
\usepackage{amsthm}
\usepackage[utf8]{inputenc}
\usepackage{nameref}
\usepackage{lineno}

\newcommand{\old}[1]{{}}
\newtheorem{theorem}{Theorem}[section]

\newtheorem{lemma}[theorem]{Lemma}

\newtheorem{obs}[theorem]{Observation}

\newcommand{\C}{{{\cal{C}}}}
\newcommand{\A}{{{\cal{A}}}}

\renewcommand{\P}{{{\cal{P}}}}

\newcommand{\I}{{{\cal{I}}}}

\newcommand{\eps}{{\varepsilon}}
\newcommand{\f}{{\varphi}}

\newcommand{\seg}[1]{\overline{#1}}

\title{The Minimum Dominating Set Problem on Bipartite Circle Graphs: Complexity and Approximation}
\author{ A. Karim Abu-Affash\thanks{Software Engineering Department, Shamoon College of Engineering, Beer-Sheva, Israel, {\tt abuaa1@sce.ac.il}.}
\and 
Paz Carmi\thanks{The Stein Faculty of Computer and Information Science, Ben-Gurion University, Beer-Sheva, Israel, {\tt carmip@bgu.ac.il}. }
\and
Joseph S. B. Mitchell\thanks{Department of Applied Mathematics and Statistics, Stony Brook University, Stony Brook, NY, USA, {\tt joseph.mitchell@stonybrook.edu}. }
}

\begin{document}

\maketitle

\begin{abstract}
A circle graph is the intersection graph of a set of chords in a circle. A dominating set of a graph $G=(V,E)$ is a subset $D\subseteq V$ such that every vertex in $V\setminus D$ is adjacent to at least one vertex of $D$. Computing a minimum dominating set is known to be NP-hard on circle graphs.

In this paper, we study the minimum dominating set problem on bipartite circle graphs, namely, circle graphs admitting a chord representation in which the chords can be partitioned into two color classes such that no two chords of the same color intersect. We prove that the problem remains NP-hard for this restricted graph class by a reduction from Planar Monotone 3-SAT. On the positive side, we present a polynomial-time 2-approximation algorithm and develop a polynomial-time approximation scheme (PTAS) based on local search.
\end{abstract}



\section{Introduction}
Given a graph $G = (V,E)$, a subset $D$ of $V$ is a \emph{dominating set} of $G$ if every vertex in $V \setminus D$ is adjacent to at least one vertex of $D$. 
The \emph{minimum dominating set} (MDS) problem is to compute a dominating set $D$ of minimum cardinality. 
The MDS problem is fundamental in the field of graph theory. The problem has real-world applications in network design, facility location, and the analysis of social networks~\cite{Wu1999, Corcoran2021, Haynes2017, Czuba2025, Kelleher1988}. It has been shown that the MDS problem is NP-hard on general graphs~\cite{Garey1978}, as are many other variants of domination, including total dominating set, connected dominating set, independent dominating set, and dominating clique~\cite{Cockayne1980, Garey1978, Goddard2013}. All of these variants, except for independent dominating set, are also NP-hard on chordal graphs~\cite{Brandstadt1987, Booth1982, Farber1982, Laskar1983, Lasker1984}. For a summary of the complexity of these variants on other classes of graphs, see Table~1 in~\cite{Corneil1984}.


A \emph{circle} graph $G=(V,E)$ is the intersection graph of a set $\C$ of chords in a circle, such that each vertex $v \in V$ uniquely corresponds to a chord, and there is an edge $(u,v) \in E$ if and only if the two chords corresponding to $u$ and $v$ intersect. $\C$ is called a \emph{chord intersection model} of the circle graph $G$. Equivalently, an \emph{interval model} of $G$ is a set $\I$ of intervals located on a horizontal line $L$, such that the vertices of $V$ uniquely correspond to the intervals of $\I$, and two vertices in $V$ are adjacent if and only if the corresponding intervals intersect, i.e., overlap, but neither contains the other. The three 
models of a circle graph (as a graph, as a set of chords, and as a set of intervals) are equivalent via linear time transformations~\cite{Gavril1973}. 
%
Thus, w.l.o.g., when specifying instances of problems, we assume the availability of the most convenient model.

Circle graphs have been extensively studied in the literature~\cite{Bousquet2014, Damian-Iordache2002, Damian2006, Gavril1973, Keil1993, Reddy2026}. 
Many problems that are NP-hard in general graphs become solvable in polynomial time when restricted to circle graphs, including maximum clique and maximum independent set~\cite{Gavril1973}, minimum feedback vertex set~\cite{Gavril2008}, and dominating clique~\cite{Keil1993}. On the other hand, some problems remain NP-hard in circle graphs, such as Hamiltonian cycle~\cite{Damaschke1989}, minimum clique cover~\cite{Keil2006}, and $k$-coloring for $k \ge 4$~\cite{Unger1988}.

Determining the complexity of the MDS problem in circle graphs was first asked by Johnson~\cite{Johnson1985} in 1985. This question remained open until Keil~\cite{Keil1993} proved in 1993 that the MDS problem is NP-complete. Later, Damian and Pemmaraju~\cite{Damian-Iordache2002} provided a $(2 + \eps)$-approximation scheme for the problem and proved that the problem is APX-hard~\cite{Damian2006}. Finally, Bousquet et al.~\cite{Bousquet2014} showed that the MDS problem in circle graphs is $W[1]$-hard parameterized by the size of the solution.

In this paper, we assume that $\C$ consists of 
chords colored red and blue, such that no two chords of the same color intersect; see  Figure~\ref{fig:example}(a). Hence, the resulting circle graph is \emph{bipartite}. Moreover, we also use the interval model. That is, the intervals in $\I$ are partitioned into two sets: a set $\I_R$ of red intervals corresponding to the red chords of $\C$ and a set $\I_B$ of blue intervals corresponding to the blue chords of $\C$. 
For ease of presentation, we represent the intervals of $\I_R$ as red curves below the line $L$ and the intervals of $\I_B$ as blue curves above $L$. An example of a chord intersection model $\C$ of a bipartite circle graph and its corresponding interval model $\I$ is given in Figure~\ref{fig:example}.
\begin{figure}[htb]
    \centering
    \includegraphics[width=0.92\textwidth]{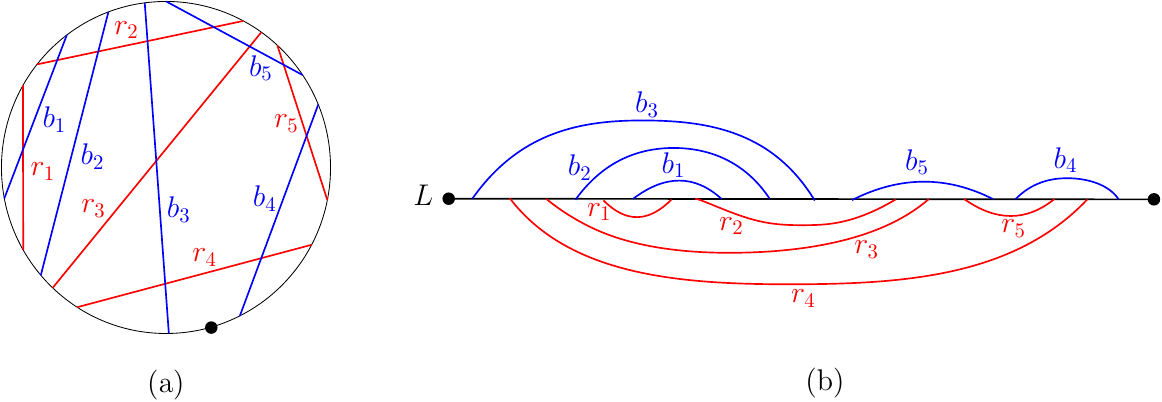}
    \caption{ (a) A chord intersection model $\C$ of a bipartite circle graph. (b) Its corresponding interval model $\I$.}
    \label{fig:example}
\end{figure}

\paragraph{Our contributions.}
We study the MDS problem on bipartite circle graphs, a natural subclass of circle graphs that admits a chord representation in which the chords can be colored red and blue so that no two chords of the same color intersect. 
We prove (in Section~\ref{sec:hardness}) that the problem remains NP-hard even under this restriction. Since the known NP-hardness construction for circle graphs (used by Keil~\cite{Keil1993}) does not preserve bipartiteness, we provide a new reduction from Planar Monotone 3-SAT. 
We also provide (in Section~\ref{sec:approxAlg}) a polynomial-time $2$-approximation algorithm. The algorithm computes an optimal set of red intervals that dominate all blue intervals and, symmetrically, an optimal set of blue intervals that dominate all red intervals. 
Finally, we show (in Section~\ref{sec:PTAS}) that the problem admits a polynomial-time approximation scheme (PTAS). We show that the standard local-search yields a PTAS for computing a minimum dominating set in bipartite circle graphs. 

\section{NP-hardness} \label{sec:hardness}
In this section, we prove that the MDS problem on bipartite circle graphs is NP-hard. 
More precisely, we prove that, given a bipartite circle graph $G=(V,E)$ and an integer $k$, determining whether there exists a dominating set of $G$ of size at most $k$ is NP-hard. The proof is by a reduction from \emph{Planar Monotone 3-SAT}, which is known to be NP-hard~\cite{deBerg2012}.

A 3-SAT formula is \emph{monotone} if each clause consists of only three positive or only three negative literals. The corresponding graph of 3-SAT formula $\f$ is a bipartite graph $G_\f$, such that each variable and each clause in $\f$ is a vertex in $G_\f$ and there is an edge between a variable $x_i$ and a clause $C_j$ if and only if $x_i$ or its negation $\seg{x_i}$ appears in $C_j$. 
If $G_\f$ is planar, then $\f$ is called a planar 3-SAT formula. 

Given a planar monotone 3-SAT formula $\f$, a \emph{monotone rectilinear representation} of $\f$ is a rectilinear representation where (i) the variables are drawn as axis-aligned rectangles arranged along a horizontal line $l$, (ii) all positive clauses are drawn as axis-aligned rectangles and lie on one side (above) of $l$, (iii) all negative clauses are drawn as axis-aligned rectangles and lie on the other side (below) of $l$, and (iv) the edges connecting the variables to the clauses are drawn as vertical segments. See Figure~\ref{fig:PM3SAT} for an example.   
\begin{figure}[htb]
    \centering
    \includegraphics[width=0.65\textwidth]{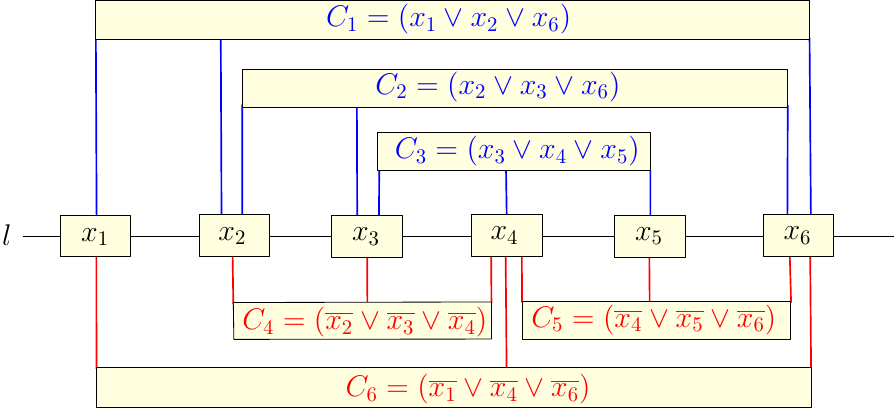}
    \caption{ An example of a monotone rectilinear representation of a planar monotone 3-SAT formula.}
    \label{fig:PM3SAT}
\end{figure}

Let $\f$ be a planar monotone 3-SAT formula defined on the variables $x_1,\dots,x_n$ and the clauses $C_1,\dots,C_m$, with its monotone rectilinear representation. We show how to construct, in polynomial time, a set of bicolored intervals $\I$ on a horizontal line $L$ (corresponding to a bipartite circle graph $G$) and fix an integer $k$, such that $G$ has a dominating set of size $k$ if and only if $\f$ is satisfiable. Let $L$ be a horizontal line. We draw the blue intervals as curves above $L$ and the red intervals as curves below $L$.

\paragraph{Variable gadget.} 
For each variable $x_i$, we create a set $A_i$ consisting of 10 intervals: five red intervals $r_i^1,r_i^2,\dots, r_i^5$ arranged from left to right on $L$ and five blue intervals $b_i^1,b_i^2,\dots, b_i^5$ arranged from right to left on $L$, such that (i) $r_i^2$ is contained in $r_i^3$, and $r_i^1, r_i^3, r_i^4, r_i^5$ are pairwise disjoint, (ii) $b_i^2$ is contained in $b_i^3$, and $b_i^1, b_i^3, b_i^4, b_i^5$ are pairwise disjoint, and (iii) $(r_i^1, b_i^5, r_i^2, b_i^4, r_i^4, b_i^2, r_i^5, b_i^1)$ forms a simple path in $G$; see Figure~\ref{fig:variable}.
\begin{figure}[htb]
    \centering
    \includegraphics[width=0.45\textwidth]{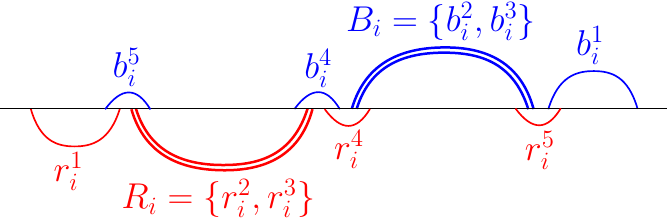}
    \caption{ A gadget of a variable $x_i$.}
    \label{fig:variable}
\end{figure}

Let $R_i=\{r_i^2,r_i^3\}$ and $B_i=\{b_i^2,b_i^3\}$. Notice that each set of intervals $A_i$ corresponding to a variable $x_i$ can be optimally dominated by three intervals from $A_i$. 
Moreover, any minimum dominating set $D_i$ of $A_i$ must satisfy: (i) $D_i$ contains only one of the intervals $r_i^1,b_i^5$, only one of the intervals $r_i^4,b_i^4$, and only one of the intervals $r_i^5,b_i^1$, (ii) $D_i$ does not contain any of the intervals of $R_i \cup B_i$, (iii) if $r_i^1 \in D_i$, then $b_i^1 \notin D_i$, and (iv) if $r_i^4 \in D_i$, then $b_i^4 \notin D_i$. This follows by a straightforward case analysis on the constant-size gadget in Figure~\ref{fig:variable}.
Therefore, we have the following observation.
\begin{obs}\label{obs:var}
$A_i$ can be optimally dominated by only one of the sets $\{r_i^1,b_i^4,r_i^5\}$, $\{r_i^4,r_i^5,b_i^5\}$, $\{b_i^4,b_i^5,r_i^5\}$, or $\{b_i^1,r_i^4,b_i^5\}$.
\end{obs}

In our construction, we show that there exists a minimum dominating set of $\I$ that includes only one of the sets $\{r_i^1,b_i^4,r_i^5\}$ or $\{b_i^1,r_i^4,b_i^5\}$.  
Thus, we associate the set $\{r_i^1,b_i^4,r_i^5\}$ with the assignment $x_i = T$ and the set $\{b_i^1,r_i^4,b_i^5\}$ with the assignment $x_i = F$.

We arrange (the intervals of) the gadgets $A_1, A_2,\dots, A_n$ on $L$, such that $R_n \subset R_{n-1} \subset \dots \subset R_2 \subset R_1$, $B_1 \subset B_2 \subset \dots \subset B_n$, and the intervals of every two gadgets $A_i$ and $A_j$ can intersect each other only with $R_i, B_i$ and $R_j, B_j$; see Figure~\ref{fig:variables}.
\begin{figure}[htb]
    \centering
    \includegraphics[width=0.78\textwidth]{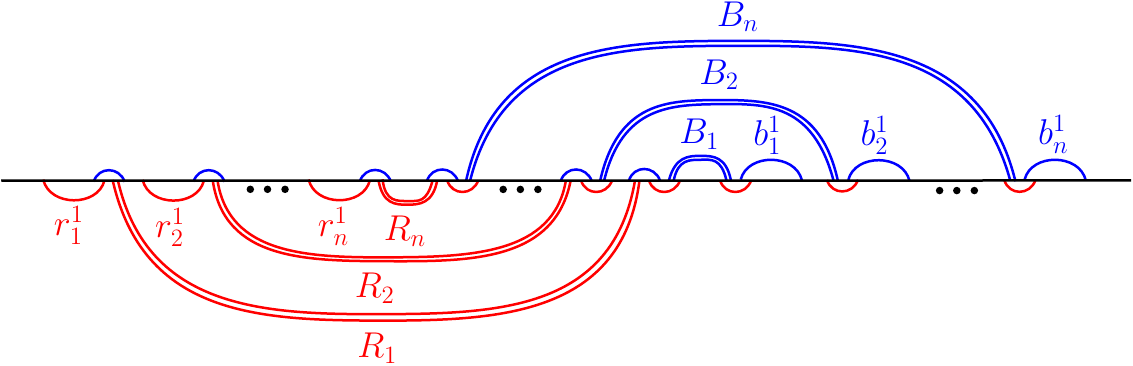}
    \caption{ Arranging the variables gadgets on $L$.}
    \label{fig:variables}
\end{figure}

\paragraph{Positive clause gadgets.} For each positive clause $C_j=(x_p \vee x_q \vee x_t)$, we create two disjoint blue intervals $b_j^{p,q},b_j^{q,t}$. Moreover, we create a set $R'_j$ containing two red intervals $r'_j, r''_j$, such that $r'_j$ is contained in $r''_j$ and they intersect both $b_j^{p,q}$ and $b_j^{q,t}$; see Figure~\ref{fig:clause}(a).
 \begin{figure}[htb]
    \centering
    \includegraphics[width=0.8\textwidth]{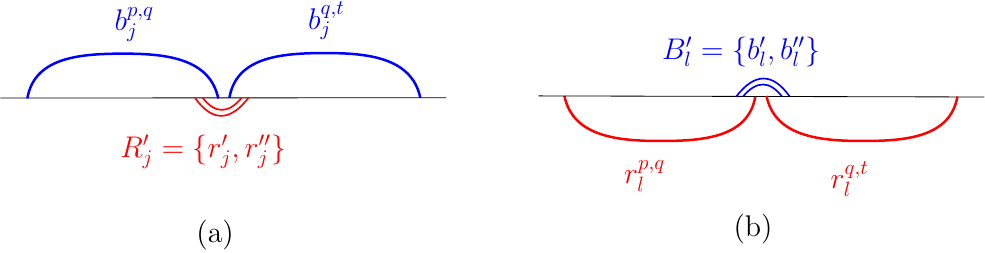}
    \caption{(a) A gadget of a positive clause $C_j=(x_p \vee x_q \vee x_t)$. (b) A gadget of a negative clause $C_l=(\seg{x_p} \vee \seg{x_q} \vee \seg{x_t})$.}
    \label{fig:clause}
\end{figure}

\paragraph{Negative clause gadgets.} The gadget of each negative clause $C_l=(\seg{x_p} \vee \seg{x_q} \vee \seg{x_t})$ is created symmetrically; see Figure~\ref{fig:clause}(b).

\paragraph{Connection between clauses and variables.} The connection between a positive clause $C_j=(x_p \vee x_q \vee x_t)$ and the gadgets corresponding to the variables $x_p, x_q, x_t$ is established via the set $R'_j$; as shown in Figure~\ref{fig:connection}(a). More precisely, We locate the blue intervals $b_j^{p,q},b_j^{q,t}$ of $C_j$ along $L$, such that 
\begin{itemize}
    \item the two red intervals of $R'_j$ are contained in the interval $r_q^1$ of $x_q$;
    \item the left endpoint of $b_j^{p,q}$ is located inside the interval $r_p^1$ of $x_p$; and
    \item the right endpoint of $b_j^{q,t}$ is located inside the interval $r_t^1$ of $x_t$.
 \end{itemize}
 \begin{figure}[htb]
    \centering
    \includegraphics[width=0.88\textwidth]{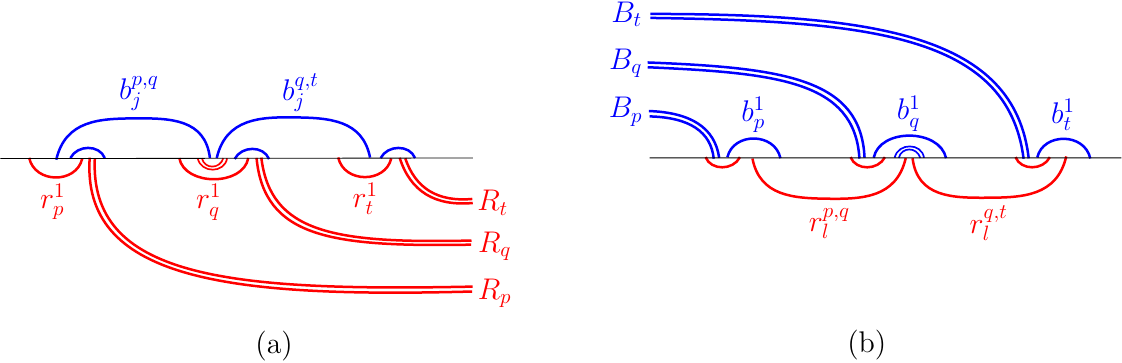}
    \caption{The connections of (a) a positive clause $C_j=(x_p \vee x_q \vee x_t)$ and (b) a negative clause $C_l=(\seg{x_p} \vee \seg{x_q} \vee \seg{x_t})$ with the gadgets of the variables $x_p, x_q, x_t$.}
    \label{fig:connection}
\end{figure}

The connection between a negative clause and the gadgets of its variables is made symmetrically, as shown in Figure~\ref{fig:connection}(b). Figure~\ref{fig:reduction} shows the construction of $\I$ corresponding to the formula in Figure~\ref{fig:PM3SAT}.
\begin{figure}[htb]
    \centering
    \includegraphics[width=0.77\textwidth]{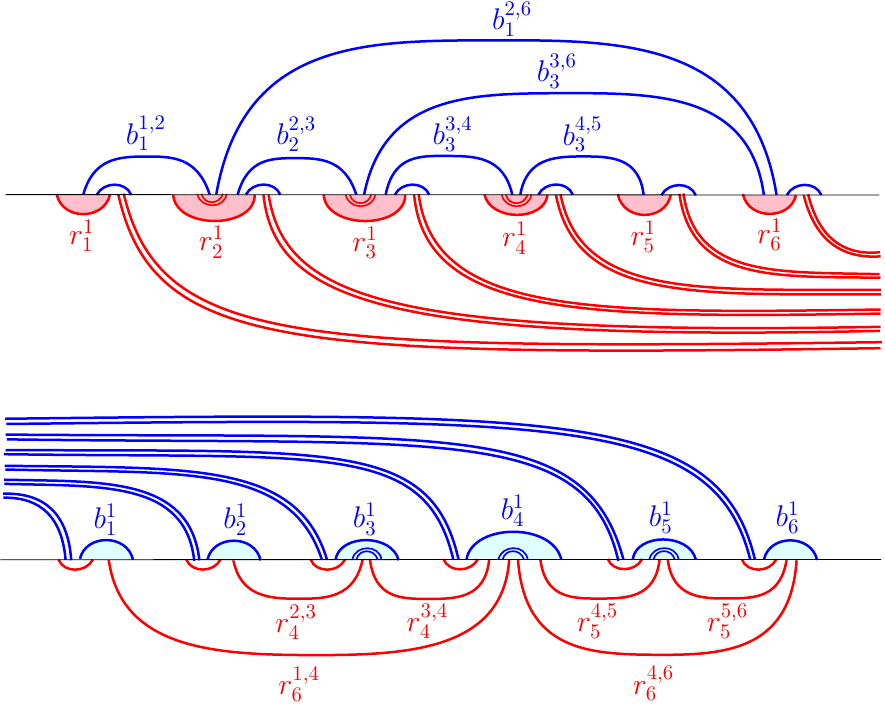}
    \caption{The intervals created by the reduction from the formula in Figure~\ref{fig:PM3SAT}.}
    \label{fig:reduction}
\end{figure}

Clearly, the final set $\I$ contains $10n + 4m$ intervals and is constructed in polynomial time. 
Moreover, the intervals in $\I$ are partitioned into red and blue, and no two intervals of the same color intersect. 
Hence, the intersection graph of the intervals in $\I$ is bipartite.
Let $G$ be the bipartite circle graph corresponding to $\I$.
%

\begin{lemma}\label{lemma:properties}
There exists a minimum dominating set $D$ of $G$ satisfying the following: 
\begin{enumerate}
    \item For each variable $x_i$, $D$ contains one of the triples $\{r_i^1,b_i^4,r_i^5\}$ or $\{b_i^1,r_i^4,b_i^5\}$ from each set $A_i$.
    \item For each positive clause $C_j=(x_p \vee x_q \vee x_t)$, $D$ contains at least one of the intervals $b_j^{p,q}$, $b_j^{q,t}$.
    \item For each negative clause $C_l=(\seg{x_p} \vee \seg{x_q} \vee \seg{x_t})$, $D$ contains at least one of the intervals $r_l^{p,q}$, $r_l^{q,t}$.
\end{enumerate}
\end{lemma}
\begin{proof}
The proof follows from the construction of $\I$.
\begin{enumerate}
    \item For each variable $x_i$, the intervals $r_i^4, r_i^5, b_i^4, b_i^5$ of $A_i$ are independent of the other intervals of $\I$, except for the intervals $r_i^1, r_i^2, r_i^3, b_i^1, b_i^2, b_i^3$. Thus, $D$ must contain three intervals from each set $A_i$, and Observation~\ref{obs:var} still holds. That is, $D$ must contain one of the sets $\{r_i^1,b_i^4,r_i^5\}$, $\{r_i^4,r_i^5,b_i^5\}$, $\{b_i^4,b_i^5,r_i^5\}$, or $\{b_i^1,r_i^4,b_i^5\}$. Therefore, if $D$ contains one of the sets $\{r_i^4,r_i^5,b_i^5\}$ or $\{b_i^4,b_i^5,r_i^5\}$, then we can replace it with one of the sets $\{r_i^1,b_i^4,r_i^5\}$ or $\{b_i^1,r_i^4,b_i^5\}$ without increasing the size of $D$. 
    %
    \item For each positive clause $C_j=(x_p \vee x_q \vee x_t)$, the intervals in the sets $R'_j$ do not intersect any of the other intervals of $\I$ except $b_j^{p,q}$ and $b_j^{q,t}$. Thus, if $D$ contains one of the intervals of $R'_j$, then we can replace it with $b_j^{p,q}$ or with $b_j^{q,t}$ without increasing the size of $D$.
    Therefore, $D$ contains at least one of the intervals $b_j^{p,q}, b_j^{q,t}$ to dominate the intervals of $R'_j$.
    \item For each negative clause $C_l=(\seg{x_p} \vee \seg{x_q} \vee \seg{x_t})$, the intervals in the sets $B'_l$ do not intersect any of the other intervals of $\I$ except $r_l^{p,q}$ and $r_l^{q,t}$. Thus, if $D$ contains one of the intervals of $B'_l$, then we can replace it with $r_l^{p,q}$ or with $r_l^{q,t}$ without increasing the size of $D$. 
    Therefore, $D$ contains at least one of the intervals $r_l^{p,q}, r_l^{q,t}$ to dominate the intervals of $B'_l$.
\end{enumerate}
\end{proof}

It remains to fix the integer $k$. By Lemma~\ref{lemma:properties}, any minimum dominating set of $G$ must contain at least three intervals to dominate each set $A_i$, at least one interval to dominate each set $R'_j$, and at least one interval to dominate each set $B'_l$.  
Thus, we set $k = 3n + m$.

\begin{lemma}
     $\f$ is satisfiable if and only if $G$ has a dominating set of size $k$.
\end{lemma}
\begin{proof}
Suppose that $\f$ is satisfiable, and let $\A$ be a truth assignment for $\f$. Since $\A$ satisfies $\f$, each clause contains at least one literal with a true value.
We construct a dominating set $D$ of $ G$ as follows.
\begin{enumerate}
    \item[(1)] For each variable $x_i$, such that $x_i = T$, we add the intervals $r_i^1$, $b_i^4$, and $r_i^5$ to $D$. These intervals dominate all the intervals of the set $A_i$. 
    \item[(2)]For each variable $x_i$, such that $x_i = F$, we add the intervals $b_i^1$, $r_i^4$, and $b_i^5$ to $D$. These intervals dominate all the intervals of the set $A_i$. 
    \item[(3)] For each positive clause $C_j=(x_p \vee x_q \vee x_t)$, at least one of the literals $x_p$, $x_q$, and $x_t$ is assigned true (if there is more than one, we select an arbitrary one). Assume, w.l.o.g., that $x_p =T$.  Thus, $r_p^1 \in D$. We add the interval $b_j^{q,t}$ to $D$. Therefore, $r_p^1$ dominates $b_j^{p,q}$, and $b_j^{q,t}$ dominates the two intervals of the set $R'_j$. (If $x_q = T$, then $r_q^1 \in D$ dominates both $b_j^{p,q}$ and $b_j^{q,t}$. In this case, it does not matter which one of them is added to $D$.) 
    \item[(4)] For each negative clause $C_l=(\seg{x_p} \vee \seg{x_q} \vee \seg{x_t})$, at least one of the literals $\seg{x_p}$, $\seg{x_q}$, and $\seg{x_t}$ is assigned true (if there is more than one, we select an arbitrary one). Assume, w.l.o.g., that $\seg{x_p} =T$. Thus, $b_p^1 \in D$. We add the interval $r_l^{q,t}$ to $D$. Therefore, $b_p^1$ dominates $r_l^{p,q}$, and $r_l^{q,t}$ dominates the two intervals of the set $B'_l$. (If $\seg{x_q} = T$, then $b_q^1 \in D$ dominates both $r_l^{p,q}$ and $r_l^{q,t}$. In this case, it does not matter which one of them is added to $D$.)
\end{enumerate}
Clearly, $D$ contains $k = 3n + m$ intervals. Moreover, by (1) and (2), $D$ dominates all the intervals in the sets $A_1, A_2,\dots, A_n$, by (1) and (3), $D$ dominates all the intervals in each gadget corresponding to a positive clause, and by (2) and (4), $D$ dominates all the intervals in each gadget corresponding to a negative clause. Thus, $D$ dominates all the intervals in $\I$, and therefore, $G$ has a dominating set of size $k$.

\medskip
Conversely, suppose that $G$ has a dominating set $D$ of size $k$ satisfying the properties in Lemma~\ref{lemma:properties}. Thus, $D$ contains at least three intervals from each set $A_i$ corresponding to a variable $x_i$, at least one interval from each pair $b_j^{p,q}, b_j^{q,t}$ corresponding to a positive clause $C_j$, and at least one interval from each pair $r_l^{p,q}, r_l^{q,t}$ corresponding to a negative clause $C_l$. 
Since $k = 3n + m$, $D$ contains exactly one of the triples $\{r_i^1,b_i^4,r_i^5\}$ or $\{b_i^1,r_i^4,b_i^5\}$ from each set $A_i$ corresponding to a variable $x_i$, exactly one interval from each pair $b_j^{p,q}, b_j^{q,t}$ corresponding to a positive clause $C_j$, and exactly one interval from each pair $r_l^{p,q}, r_l^{q,t}$ corresponding to a negative clause $C_l$. 
We assign $x_i = T$ if $D$ contains the triple $\{r_i^1,b_i^4,r_i^5\}$; otherwise, we assign $x_i = F$.

We show that this assignment satisfies $\f$. 
For each positive clause $C_j=(x_p \vee x_q \vee x_t)$, since exactly one of the two (blue) intervals $b_j^{p,q}, b_j^{q,t}$ is contained in $D$, the other interval must be dominated by some interval of $D$.  
By the construction, the only intervals outside the clause gadget that can intersect $b_j^{p,q}, b_j^{q,t}$ are the intervals $r_p^1$, $r_q^1$, $r_t^1$, and intervals from $R_p \cup R_q \cup R_t$. 
Moreover, by Observation~\ref{obs:var} and Lemma~\ref{lemma:properties}, no interval from $R_p \cup R_q \cup R_t$ belongs to $D$.
Thus, $D$ contains at least one of the intervals $r_p^1$, $r_q^1$, or $r_t^1$, and the variable corresponding to this interval is assigned true, which satisfies the clause $C_j$. 
Each negative clause is satisfied due to a similar argument. Therefore, $\f$ is satisfiable.  
\end{proof}

The following theorem summarizes the result of this section.
\begin{theorem} \label{thm:hardness}
The MDS problem on bipartite circle graphs is NP-hard.
\end{theorem}

\section{Approximation Algorithm} \label{sec:approxAlg}
Let $\I=\I_R \cup \I_B$ be a set of intervals arranged along a horizontal line $L$ and corresponding to a bipartite circle graph $G$, such that $\I_R$ is a set of $n$ red intervals, and $\I_B$ is a set of $m$ blue intervals. 
Let $D^*$ be a minimum dominating set of $G$, and let $|D^*|$ denote its size. 
In this section, we present a polynomial-time approximation algorithm that computes a dominating set $D$ of $G$, such that $|D| \le 2\cdot|D^*|$. 
The idea of the algorithm is to compute a minimum dominating set $R$ that dominates all the blue intervals of $\I_B$ and a minimum dominating set $B$ that dominates all the red intervals of $\I_R$, and to take $D = R \cup B$. In the following, we show how to compute the set $R$. Computing the set $B$ is symmetric.

\subsection{Computing the set $\boldsymbol{R}$} \label{sec:R}
In this section, we present a dynamic programming algorithm that computes a minimum dominating set $R$ of the blue intervals of $\I_B$. 
Notice that if there are blue intervals that do not intersect any red interval from $I_R$, then they must be in any dominating set of $\I_B$. 
Therefore, we may assume that each blue interval intersects at least one red interval, and thus $R$ contains only red intervals.

Let $N=n+m$, and let $\P=\{p_1, p_2, \dots, p_{2N}\}$ be the set of endpoints of the intervals in $\I$ listed from left to right.
We assume, w.l.o.g., that the points in $\P$ are distinct and located on consecutive integers $1,2,\dots,2N$. 
For an interval $I \in \I$, let $s(I)$ and $e(I)$ denote the starting and ending points (i.e., the left and right endpoints) of $I$, respectively, and let $C(I)$ denote the set of intervals in $\I$ that intersect $I$, i.e., $C(I) = \{I' \in \I:s(I') < s(I) < e(I') < e(I) \text{ or } s(I) < s(I') < e(I) < e(I')\}$.

For each $1 \le i \le j \le 2N$, let $\P[i,j]=\{p_i,p_{i+1},\dots,p_j\}$ be the subset of $\P$ containing the points between $p_i$ and $p_j$. 
Let $\I_B[i,j]$ denote the set of blue intervals in $\I_B$ that have both endpoints in $\P[i,j]$, i.e., $\I_B[i,j] = \{b \in \I_B:i \le s(b) < e(b) \le j\}$.
Let $\I^+_R[i,j]$ denote the set of red intervals in $\I_R$ whose ending points are in $\P[i,j]$, i.e., $\I^+_R[i,j] = \{r \in \I_R:i \le e(r) \le j\}$.

\old{
Let $b_1, b_2, \dots, b_m$ be the blue intervals in $\I_B$ sorted from left to right, such that $s(b_1) < s(b_2) < \dots < s(b_m)$.
Similarly, let $r_1, r_2, \dots, r_n$ be the blue intervals in $\I_R$ sorted from left to right, such that $s(r_1) < s(r_2) < \dots < s(r_n)$.
Let $R$ be a minimum dominating set of $\I_B$, and, for each $1 \le i \le j \le 2N$, let $R[i,j]$ be the subset of $R$ dominating $\I_B[i,j]$.
Since $R$ is a dominating set of $\I_B$, $R \cap C(b_1) \neq \emptyset$. 
Let $r_k$ be the leftmost (red) interval in $R \cap C(b_1)$ and consider the blue intervals in $\I_B \setminus C(r_k)$. We distinguish between two cases.
\medskip
\\ {\bf Case~1:} $\boldsymbol{s(b_1) < s(r_k) < e(b_1)}$. 
In this case, the blue intervals in $\I_B \setminus C(r_k)$ are partitioned into three disjoint sets: $\I_B[s(b_1)+1, s(r_k)-1]$, $\I_B[s(r_k)+1,e(r_k)-1]$, and $\I_B[e(r_k)+1,2N]$; see Figure~\ref{fig:DP1}. 
Clearly, since the red intervals in $\I_R$ do not intersect each other, we have 
\[
\begin{aligned}
R[s(b_1)+1, s(r_k)-1] \cap R[s(r_k)+1,e(r_k)-1] = \emptyset
\end{aligned}
\]
and 
\[
\begin{aligned}
R[s(r_k)+1,e(r_k)-1] \cap R[e(r_k)+1,2N] = \emptyset.
\end{aligned}
\]
Moreover, since $r_k$ is the leftmost interval in $R \cap C(b_1)$, we also have 
\[
\begin{aligned}
R[s(b_1)+1, s(r_k)-1] \cap R[e(r_k)+1,2N] = \emptyset.
\end{aligned}
\]
However, there might be an interval $r \in R[e(r_k)+1,2N]$, such that $s(r) < s(b_1)$ and $e(r) > e(r_k)$, i.e., $r$ starts before $b_1$ and ends after $r_k$. 
Hence, each interval $r \in R[s(b_1)+1, s(r_k)-1]$ has $e(r) \in \P[s(b_1)+1, s(r_k)-1]$, each interval $r \in R[s(r_k)+1,e(r_k)-1]$ has $e(r) \in \P[s(r_k)+1,e(r_k)-1]$, and each interval $r \in R[e(r_k)+1,2N]$ has $e(r) \in \P[e(r_k)+1,2N]$. 
Thus, the sets $R[s(b_1)+1, s(r_k)-1]$, $R[s(r_k)+1,e(r_k)-1]$, and $R[e(r_k)+1,2N]$ are disjoint,  
\[
\begin{aligned}
R[s(b_1)+1, s(r_k)-1] & \subseteq \I^+_R[s(b_1)+1, s(r_k)-1], \\
R[s(r_k)+1,e(r_k)-1] & \subseteq \I^+_R[s(r_k)+1,e(r_k)-1], \text{ and} \\
R[e(r_k)+1,2N] & \subseteq \I^+_R[e(r_k)+1,2N],
\end{aligned}
\]
and, 
\[
\begin{aligned}
O[s(b_1)+1, s(r_k)-1] & = |R[s(b_1)+1, s(r_k)-1]|, \\
O[s(r_k)+1,e(r_k)-1] & = |R[s(r_k)+1,e(r_k)-1]|, \text{ and} \\
O[e(r_k)+1,2N] & = |R[e(r_k)+1,2N]|,
\end{aligned}
\]
and therefore, 
\[
\begin{aligned}
|R| = O[1,2N] = O[s(b_1)+1, s(r_k)-1]   + O[s(r_k)+1,e(r_k)-1] + O[e(r_k)+1,2N] + 1 .
\end{aligned}
\]
\begin{figure}[htb]
    \centering
    \includegraphics[width=0.95\textwidth]{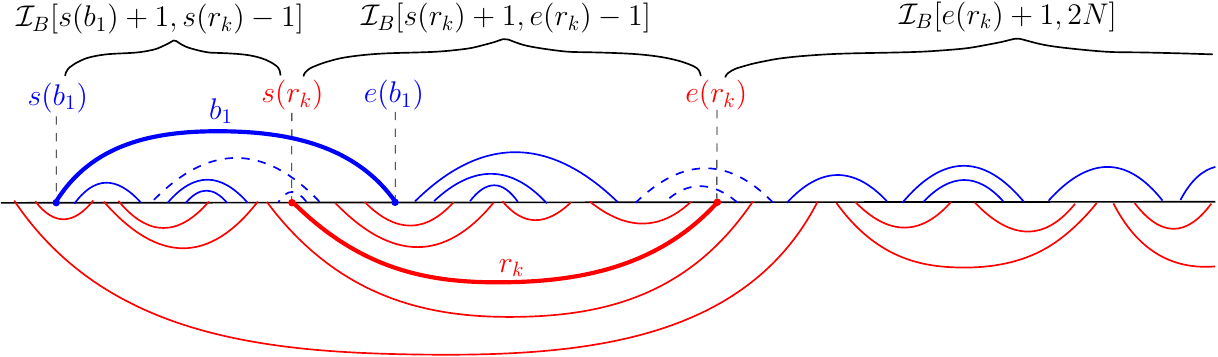}
    \caption{Case~1: $s(b_1) < s(r_k) < e(b_1)$. $r_k$ partitions the blue intervals in $\I_B \setminus C(r_k)$ into three disjoint sets: $\I_B[s(b_1)+1, s(r_k)-1]$, $\I_B[s(r_k)+1,e(r_k)-1]$, and $\I_B[e(r_k)+1,2N]$.}
    \label{fig:DP1}
\end{figure}

If $e(r_k) = 2N$, then the blue intervals in $\I_B \setminus C(r_k)$ are partitioned into two disjoint sets: $\I_B[s(b_1)+1, s(r_k)-1]$ and $\I_B[s(r_k)+1,e(r_k)-1]$. 
In this case, we have 
\[
\begin{aligned}
|R| = O[1,2N] = O[s(b_1)+1, s(r_k)-1] + O[s(r_k)+1,e(r_k)-1] + 1 .
\end{aligned}
\]
%
\\ {\bf Case~2:} $\boldsymbol{s(b_1) < e(r_k) < e(b_1)}$. 
In this case, the blue intervals in $\I_B \setminus C(r_k)$ are partitioned into two disjoint sets: $\I_B[s(b_1)+1, e(r_k)-1]$ and $\I_B[e(r_k)+1,2N]$; see Figure~\ref{fig:DP2}. 
As in Case~1, the red intervals in $\I_R$ do not intersect each other, we have 
$R[s(b_1)+1, e(r_k)-1] \cap R[e(r_k)+1,2N] = \emptyset$, and there might be an interval $r \in R[e(r_k)+1,2N]$, such that $s(r) < s(r_k)$ and $e(r) > e(r_k)$, i.e., $r$ starts before $r_k$ and ends after $r_k$.
Hence, each interval $r \in R[s(b_1)+1, e(r_k)-1]$ has $e(r) \in \P[s(b_1)+1, e(r_k)-1]$ and each interval $r \in R[e(r_k)+1,2N]$ has $e(r) \in \P[e(r_k)+1,2N]$.
Thus, the sets $R[s(b_1)+1, e(r_k)-1]$ and $R[e(r_k)+1,2N]$ are disjoint, 
\[
\begin{aligned}
R[s(b_1)+1, e(r_k)-1] & \subseteq \I^+_R[s(b_1)+1, e(r_k)-1], \text{ and} \\
R[e(r_k)+1,2N] & \subseteq \I^+_R[e(r_k)+1,2N].
\end{aligned}
\]
and, 
\[
\begin{aligned}
O[s(b_1)+1, e(r_k)-1] & = |R[s(b_1)+1, s(r_k)-1]|, \text{ and} \\
O[e(r_k)+1,2N] & = |R[e(r_k)+1,2N]|,
\end{aligned}
\]
and therefore, 
\[
\begin{aligned}
|R| = O[1,2N] = O[s(b_1)+1, e(r_k)-1]  + O[e(r_k)+1,2N] + 1 .
\end{aligned}
\]
\begin{figure}[htb]
    \centering
    \includegraphics[width=0.95\textwidth]{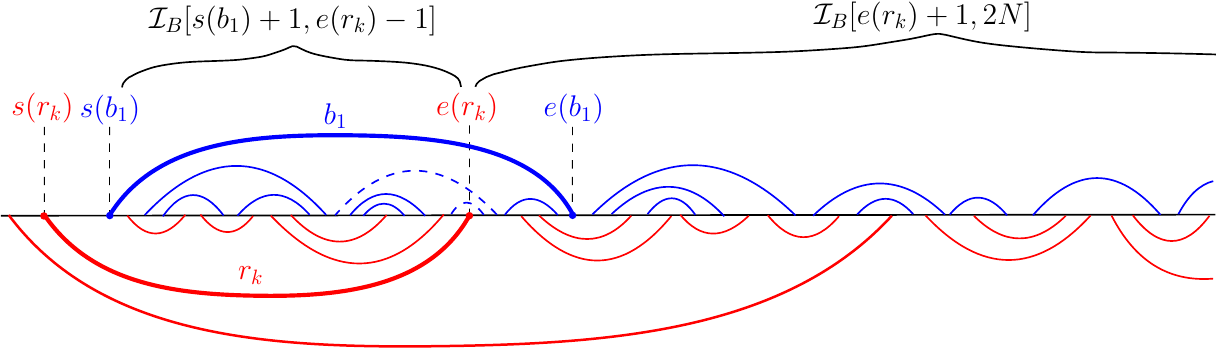}
    \caption{Case~2: $s(b_1) < e(r_k) < e(b_1)$. $r_k$ partitions the blue intervals in $\I_B \setminus C(r_k)$ into two disjoint sets: $\I_B[s(b_1)+1, e(r_k)-1]$ and $\I_B[e(r_k)+1,2N]$.}
    \label{fig:DP2}
\end{figure}

Notice that, in both cases, the subproblems obtained by $r_k$ are independent. 
Therefore, in order to compute $|R| = O[1,2N]$, for each interval $r_k \in C(b_1)$, we compute $O[s(b_1)+1, s(r_k)-1] + O[s(r_k)+1,e(r_k)-1] + O[e(r_k)+1,2N] + 1$, in case $s(b_1) < s(r_k) < e(b_1)$, or $O[s(b_1)+1, e(r_k)-1] + O[e(r_k)+1,2N] + 1$, in case $s(b_1) < e(r_k) < e(b_1)$, and take the minimum over these values.
}

\paragraph{Subproblems.}
For each $1 \le i \le j \le 2N$, a subproblem is defined on the points in $\P[i,j]$. 
More precisely, a subproblem is to compute the minimum number of red intervals from $\I^+_R[i,j]$ that dominate the blue intervals in $\I_B[i,j]$.
For each subproblem $\P[i,j]$, we define $O[i,j]$ to be the minimum number of red intervals from $\I^+_R[i,j]$ that dominate the blue intervals in $\I_B[i,j]$. 
If $i=j$ or $\I_B[i,j] = \emptyset$, then we set $O[i,j] = 0$.   

Let $b=LB[i,j]$ be the blue interval in $\I_B[i,j]$ having the leftmost starting point $s(b)$. 
For each $r \in C(b)$ such that $e(r) \le j$, we define a value $f(r)$ as follows:

\begin{itemize}
    \item \textbf{If} $s(b) < s(r) < e(b)$, \textbf{then} 
        \begin{itemize}
            \item[] \textbf{if} $e(r) < j$, \textbf{then}
                \[
                    \begin{aligned}
                    f(r) = O[s(b)+1, s(r)-1] + 
                              O[s(r)+1,e(r)-1] +  
                              O[e(r)+1,j] + 1,
                    \end{aligned}
                \]
                \item[] \textbf{else  \  ($e(r) = j$)}
                \[
                    \begin{aligned}
                        f(r) = O[s(b)+1, s(r)-1] +  
                          O[s(r)+1,e(r)-1]+1.   
                    \end{aligned}
    \]
        \end{itemize}

    \item \textbf{If} $s(b) < e(r) < e(b)$, \textbf{then} 
    \[
        \begin{aligned}
            f(r) = O[s(b)+1, e(r)-1] +    
                     O[e(r)+1,j] + 1.
        \end{aligned}
    \]
\end{itemize}

The following lemma establishes the optimal substructure property of $O[i,j]$. 
\begin{lemma} \label{lemma:DP}
For each subproblem $\P[i,j]$, $O[i,j]$ satisfies the following recurrence:
\[
O[i,j] =  
\begin{cases} 
\quad 0  &:\text{if } i=j \text{ or }  \I_B[i,j] = \emptyset; \\
\underset{\underset{ e(r) \le j}{r \in C(LB[i,j])}}{\min} f(r)  &:\text{otherwise.} 
\end{cases}
\]
\end{lemma}
\begin{proof}
If $i=j$ or $\I_B[i,j] = \emptyset$, then there are no blue intervals to dominate, and hence $O[i,j]=0$. 
Otherwise, let $b=LB[i,j]$ and let $r$ be the leftmost red interval in $C(b)$ (i.e., having the leftmost starting point $s(r)$), such that $f(r)$ is minimized over all $r \in C(b)$ with $e(r) \le j$. We distinguish between two cases.
\medskip
\\ {\bf Case~1:} $\boldsymbol{s(b) < s(r) < e(b)}$. We distinguish between two subcases according to $e(r)$.
\\ {\bf Case~1.1:} $\boldsymbol{e(r) < j}$. In this case, the blue intervals in $\I_B[i,j] \setminus C(r)$ are partitioned into three disjoint sets: $\I_B[s(b)+1, s(r)-1]$, $\I_B[s(r)+1,e(r)-1]$, and $\I_B[e(r)+1,j]$; see Figure~\ref{fig:DP11}. 
Let $R[s(b)+1, s(r)-1]$, $R[s(r)+1,e(r)-1]$, and $R[e(r)+1,j]$ be minimum dominating sets of the blue intervals in $\I_B[s(b)+1, s(r)-1]$, $\I_B[s(r)+1,e(r)-1]$, and $\I_B[e(r)+1,j]$, respectively.

Clearly, since the red intervals in $\I_R$ do not intersect each other, we have 
\[
\begin{aligned}
R[s(b)+1, s(r)-1] \cap R[s(r)+1,e(r)-1] = \emptyset \ \ \text{ and } \  \
R[s(r)+1,e(r)-1] \cap R[e(r)+1,j] = \emptyset.
\end{aligned}
\]
Moreover, since $r$ is the leftmost interval in $C(b)$, we also have 
\[
\begin{aligned}
R[s(b)+1, s(r)-1] \cap R[e(r)+1,j] = \emptyset,
\end{aligned}
\]
otherwise, there is an interval $r' \in R[s(b)+1, s(r)-1] \cap R[e(r)+1,j]$, such that $s(b) < s(r') <s(r) < e(b) <e(r) < e(r')$, which implies that $r' \in C(b)$ and $s(r') < s(r)$, contradicting that $r$ is the leftmost interval in $C(b)$. 
However, there might be an interval $r' \in R[e(r)+1,j]$, such that $s(r') < s(b)$ and $e(r) < e(r') < j$, i.e., $r'$ starts before $b$ and ends after $r$. 
Hence, each interval $r' \in R[s(b)+1, s(r)-1]$ has $e(r') \in \P[s(b)+1, s(r)-1]$, each interval $r' \in R[s(r)+1,e(r)-1]$ has $e(r') \in \P[s(r)+1,e(r)-1]$, and each interval $r' \in R[e(r)+1,j]$ has $e(r') \in \P[e(r)+1,j]$. 
Thus, the sets $R[s(b)+1, s(r)-1]$, $R[s(r)+1,e(r)-1]$, and $R[e(r)+1,j]$ are disjoint,  
\[
\begin{aligned}
R[s(b)+1, s(r)-1] & \subseteq \I^+_R[s(b)+1, s(r)-1], \\
R[s(r)+1,e(r)-1] & \subseteq \I^+_R[s(r)+1,e(r)-1], \text{ and} \\
R[e(r)+1,j] & \subseteq \I^+_R[e(r)+1,j].
\end{aligned}
\]
Moreover, we have
\[
\begin{aligned}
O[s(b)+1, s(r)-1] & = |R[s(b)+1, s(r)-1]|, \\
O[s(r)+1,e(r)-1] & = |R[s(r)+1,e(r)-1]|, \text{ and} \\
O[e(r)+1,j] & = |R[e(r)+1,j]|.
\end{aligned}
\]
Therefore, $R[s(b)+1, s(r)-1] \cup R[s(r)+1,e(r)-1] \cup R[e(r)+1,j] \cup \{r\}$ is a minimum dominating set of $\I_B[i,j]$, and its size is 
\[
\begin{aligned}
O[s(b)+1, s(r)-1] + O[s(r)+1,e(r)-1] + O[e(r)+1,j] + 1 = f(r).
\end{aligned}
\]
\begin{figure}[htb]
    \centering
    \includegraphics[width=0.66\textwidth]{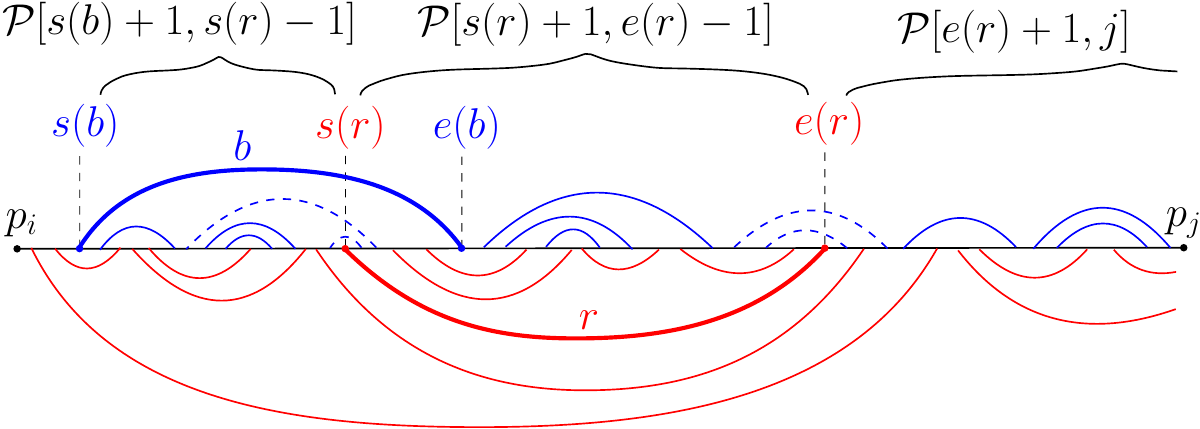}
    \caption{Case~1: $s(b) < s(r) < e(b)$. $r$ partitions the subproblem $\P[i,j]$ into three independent subproblems: $\P[s(b)+1, s(r)-1]$, $\P[s(r)+1,e(r)-1]$, and $\P[e(r)+1,j]$.}
    \label{fig:DP11}
\end{figure}
%
\\ {\bf Case~1.2:} $\boldsymbol{e(r) = j}$. In this case, the blue intervals in $\I_B[i,j] \setminus C(r)$ are partitioned into two disjoint sets: $\I_B[s(b)+1, s(r)-1]$ and $\I_B[s(r)+1,e(r)-1]$. 
Let $R[s(b)+1, s(r)-1]$ and $R[s(r)+1,e(r)-1]$ be minimum dominating sets of the blue intervals in $\I_B[s(b)+1, s(r)-1]$ and $\I_B[s(r)+1,e(r)-1]$, respectively. 

As in Case~1.1, since the red intervals in $\I_R$ do not intersect each other, we have
\[
\begin{aligned}
R[s(b)+1, s(r)-1] \cap R[s(r)+1,e(r)-1] = \emptyset, 
\end{aligned}
\]
and each interval $r' \in R[s(b)+1, s(r)-1]$ has $e(r') \in \P[s(b)+1, s(r)-1]$ and each interval $r' \in R[s(r)+1,e(r)-1]$ has $e(r') \in \P[s(r)+1,e(r)-1]$.
Thus, the sets $R[s(b)+1, s(r)-1]$ and $R[s(r)+1,e(r)-1]$ are disjoint,  
\[
\begin{aligned}
R[s(b)+1, s(r)-1] \subseteq \I^+_R[s(b)+1, s(r)-1], \ \ \text{ and } \ \
R[s(r)+1,e(r)-1] \subseteq \I^+_R[s(r)+1,e(r)-1].
\end{aligned}
\]
Moreover, 
\[
\begin{aligned}
O[s(b)+1, s(r)-1] = |R[s(b)+1, s(r)-1]| \ \ \text{ and } \ \
O[s(r)+1,e(r)-1] = |R[s(r)+1,e(r)-1]|.
\end{aligned}
\]
%
Therefore, $R[s(b)+1, s(r)-1] \cup R[s(r)+1,e(r)-1] \cup \{r\}$ is a minimum dominating set of $\I_B[i,j]$, and its size is 
\[
\begin{aligned}
O[s(b)+1, s(r)-1] + O[s(r)+1,e(r)-1] + 1 = f(r).
\end{aligned}
\]
\medskip
\\ {\bf Case~2:} $\boldsymbol{s(b) < e(r) < e(b)}$. In this case, the blue intervals in $\I_B[i,j] \setminus C(r)$ are partitioned into two disjoint sets: $\I_B[s(b)+1, e(r)-1]$ and $\I_B[e(r)+1,j]$; see Figure~\ref{fig:DP21}. 
Let $R[s(b)+1, e(r)-1]$ and $R[e(r)+1,j]$ be minimum dominating sets of the blue intervals in $\I_B[s(b)+1, e(r)-1]$ and $\I_B[e(r)+1,j]$, respectively. 

Since the red intervals in $\I_R$ do not intersect each other, we have
\[
\begin{aligned}
R[s(b)+1, e(r)-1] \cap R[e(r)+1,j] = \emptyset, 
\end{aligned}
\]
and each interval $r' \in R[s(b)+1, e(r)-1]$ has $e(r') \in \P[s(b)+1, e(r)-1]$ and each interval $r' \in R[e(r)+1,j]$ has $e(r') \in \P[e(r)+1,j]$.
Thus, the sets $R[s(b)+1, e(r)-1]$ and $R[e(r)+1,j]$ are disjoint,  
\[
\begin{aligned}
R[s(b)+1, e(r)-1] \subseteq \I^+_R[s(b)+1, e(r)-1], \ \text{ and } \ \
R[e(r)+1,j] \subseteq \I^+_R[e(r)+1,j].
\end{aligned}
\]
Moreover, 
\[
\begin{aligned}
O[s(b)+1, e(r)-1] = |R[s(b)+1, e(r)-1]| \ \ \text{ and } \ \
O[e(r)+1,j] = |R[e(r)+1,j]|.
\end{aligned}
\]
%
Therefore, $R[s(b)+1, e(r)-1] \cup R[e(r)+1,j] \cup \{r\}$ is a minimum dominating set of $\I_B[i,j]$, and its size is 
\[
\begin{aligned}
O[s(b)+1, s(r)-1] + O[s(r)+1,e(r)-1] + 1 = f(r).
\end{aligned}
\] 
\begin{figure}[htb]
    \centering
    \includegraphics[width=0.65\textwidth]{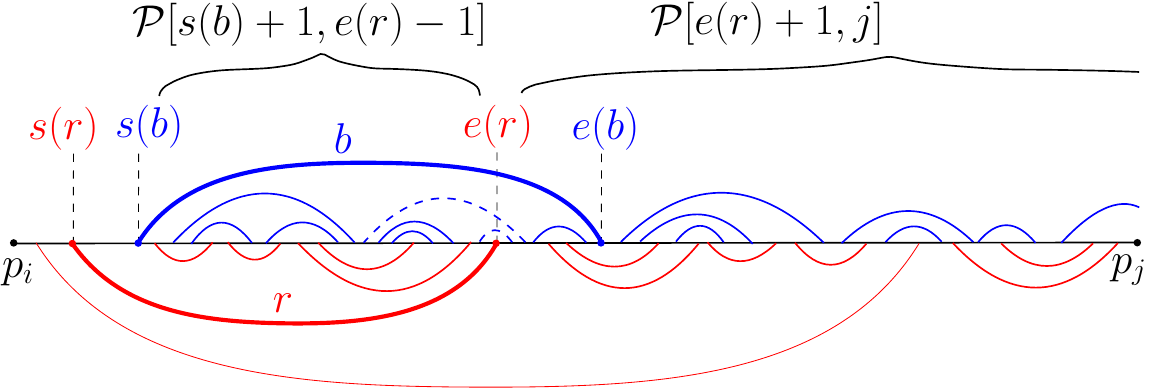}
    \caption{Case~2: $s(b) < e(r) < e(b)$. $r$ partitions the subproblem $\P[i,j]$ into two independent subproblems: $\P[s(b)+1, e(r)-1]$ and $\P[e(r)+1,j]$.}
    \label{fig:DP21}
\end{figure}
\medskip
\\ Since $r$ minimizes $f(r)$, we have $O[i,j] = f(r)$, and therefore $O[i,j] = \underset{\underset{ e(r) \le j}{r \in C(b)}}{\min} f(r)$.
\end{proof}

According to Lemma~\ref{lemma:DP}, our goal is to compute $O[1,2N] = |R|$.

\subsubsection{Algorithm implementation and running time}
We focus on computing $O[1,2N] = |R|$. By slightly modifying the algorithm using a standard backtracking technique, an actual minimum dominating set $R$ can also be recovered within asymptotically the same time complexity.

We compute $O[1, 2N]$ using dynamic programming. 
The dynamic programming table $O$ has $2N$ rows and $2N$ columns. Each entry $O[i,j]$ stores the solution for the subproblem defined on $\P[i,j]$, i.e., the minimum number of red intervals from $\I^+_R[i,j]$ needed to dominate all blue intervals in $\I_B [i,j]$. 

We fill in the table iteratively in such a way that, for each entry $O[i,j]$, all of the entries needed to compute $O[i,j]$ have already been computed. To this end, we fill the table $O$ in increasing order of the length $j-i$. For each $1 \le i \le j \le 2N$, in increasing order of $j-i$, if $i=j$ or $\I_B[i,j] = \emptyset$, then we set $O[i,j] = 0$. Otherwise, for $b = LB[i,j]$ and for each $r \in C(b)$ with $e(r) \le j$, we compute $O[s(b)+1, s(r)-1] + O[s(r)+1,e(r)-1] + O[e(r)+1,j] + 1$, in case $s(b) < s(r) < e(b) < e(r) < j$, $O[s(b)+1, s(r)-1] + O[s(r)+1,e(r)-1] + 1$, in case $s(b) < s(r) < e(b) < e(r) = j$, or $O[s(b)+1, e(r)-1] + O[e(r)+1,j] + 1$, in case $s(r) < s(b) < e(r)$, and take the minimum over these values; see Algorithm~\ref{algo:alg1}. 
\begin{algorithm}[ht]
\caption{\textsc{ComputeO}$(\I_R,\I_B)$}
\label{algo:alg1}
\begin{algorithmic}[1]
\REQUIRE A set $\I_R$ of $n$ red intervals and a set $\I_B$ of $m$ blue intervals.
\ENSURE The minimum number of red intervals from $\I_R$ needed to dominate all intervals in $\I_B$.

\STATE $N \leftarrow n + m$.
\STATE let $1,2,\ldots,2N$ be the endpoints of $\I_R \cup \I_B$ listed from left to right.
\FOR {$i \leftarrow 1$ to $2N$} 
	\FOR {$j \leftarrow 1$ to $2N$}
        \STATE $O[i,j] \leftarrow \infty$
    \ENDFOR    
\ENDFOR

\FOR {each $b \in \I_B$}
	\STATE let $C(b)$ be the set of red intervals intersecting $b$.
\ENDFOR

\FOR{$\ell=0$ to $2N-1$}
    \FOR{$i=1$ to $2N-\ell$}
        \STATE $j\leftarrow i+\ell$
        \IF{$i=j$ or $\I_B[i,j]=\emptyset$}
            \STATE $O[i,j]\leftarrow 0$
        \ELSE
            \STATE let $b=LB[i,j]$ be the interval in $\I_B[i,j]$ with leftmost endpoint $s(b)$.
            \FOR{each $r\in C(b)$ with $e(r)\le j$}
                \IF{$s(b)<s(r)<e(b)<e(r)<j$}
                    \STATE $O[i,j]\leftarrow \min\bigl\{O[i,j] \ , \ O[s(b)+1,s(r)-1]$ $+ O[s(r)+1,e(r)-1]$ $+O[e(r)+1,j]+1\bigr\}$
                \ELSIF{$s(b)<s(r)<e(b)<e(r)=j$}
                    \STATE $O[i,j]\leftarrow \min\bigl\{O[i,j] \ , \ O[s(b)+1,s(r)-1]$ $+ O[s(r)+1,e(r)-1]+1\bigr\}$
                \ELSIF{$s(b)<e(r)<e(b)$}
                    \STATE $O[i,j]\leftarrow \min\bigl\{O[i,j] \ , \ O[s(b)+1,e(r)-1]$ $+O[e(r)+1,j]+1\bigr\}$
                \ENDIF
            \ENDFOR
        \ENDIF
    \ENDFOR
\ENDFOR

\RETURN $O[1,2N]$
\end{algorithmic}
\end{algorithm}

\paragraph{Running time.}
For each interval $b \in \I_B$, we can compute $C(b)$ in $O(n)$ time.
For each entry $O[i,j]$, we consider every red interval that intersects $LB[i,j]$. Thus, $O[i,j]$ can be computed in $O(n)$ time. Since the table $O$ has $O(N^2)$ entries, the total running time for computing $|R| = O[1,2N]$ is $O(N^2 n) = O((n+m)^2n)$.

\paragraph{Recovering $R$.}
To recover the actual set $R$, we store, for each entry $O[i,j]$, a red interval $r$ that attains the minimum. After the table has been filled, we backtrack from $O[1,2N]$ and recursively collect the selected red intervals.

\subsection{Approximating the MDS Problem}
Recall that $D^*$ denotes a minimum dominating set of $G$. 
We compute a minimum dominating set $R$ of red intervals that dominate the blue intervals in $\I_B$, using the algorithm described in Section~\ref{sec:R}. 
Using a symmetric algorithm, we compute a minimum dominating set $B$ of blue intervals that dominate the red intervals in $\I_R$.  
%
\begin{lemma} \label{lemma:approx}
$|R|\le |D^*|$ and $|B|\le |D^*|$.
\end{lemma}
\begin{proof}
Consider a minimum dominating set $D^*$. Every blue interval in $\I_B$ is either in $D^*$ or is dominated by some red interval in $D^*$. Since every blue interval is intersected by at least one red interval, we can replace each blue interval in $D^*$ with an arbitrary intersecting red interval from $\I_R$. This yields a set $R'$ of red intervals of size at most $|D^*|$ that dominates all the blue intervals. Since $R$ is a minimum dominating set of the blue intervals, we have $|R| \le |R'| \le |D^*|$. 
A symmetric argument holds for $B$.   
\end{proof}

Let $D = R \cup B$. By Lemma~\ref{lemma:approx}, $|D| \le 2\cdot |D^*|$.

\begin{theorem} \label{thm:approx}
Given a bipartite circle graph $G$ containing $n$ red and $m$ blue chords, one can compute in $O((n+m)^3)$ time a dominating set of $G$ of size at most twice the size of a minimum dominating set of $G$.     
\end{theorem}

\section{PTAS for the MDS Problem} \label{sec:PTAS}
Let $G$ be a bipartite circle graph represented by a set $\I=\I_R \cup \I_B$ of intervals arranged along a horizontal line $L$, such that $\I_R$ is a set of $n$ red intervals, and $\I_B$ is a set of $m$ blue intervals. 
In this section, we show that the standard local-search algorithm leads to a PTAS for computing a minimum dominating set of $G$. For the analysis, we use a separator-based technique introduced independently by Chan and Har-Peled~\cite{Chan2012} and by Mustafa and Ray~\cite{Mustafa2010}. The main part of this proof technique is to show the existence of a planar graph satisfying a locality condition (to be defined in Section~\ref{sec:PTASratio}).

The algorithm is given in Section~\ref{sec:PTASalg}, and the analysis of the approximation ratio is provided in Section~\ref{sec:PTASratio}.

\subsection{The algorithm} \label{sec:PTASalg}
A dominating set $D \subseteq \I$ is called $k$-\textit{locally optimal} if one cannot obtain a smaller dominating set by replacing a subset $X \subseteq D$ of size at most $k$ with a subset of size at most $|X|-1$ from $\I$. Given a constant $\eps > 0$, our algorithm computes a $k$-locally optimal dominating set for $k=\frac{c}{\eps^2}$, where $c$ is a suitably large constant. 
The algorithm begins with an arbitrary dominating set $D$ of $G$ (for example, the dominating set computed in Section~\ref{sec:approxAlg}) and proceeds by making small exchanges of size at most $k$. That is, as long as there is a subset $X \subseteq D$ of size at most $k$ and a subset $X' \subseteq \I$ of size at most $|X|-1$, such that $(D\setminus X) \cup X'$ is a dominating set of $G$, we replace $D$ with $(D\setminus X) \cup X'$. The algorithm stops when no further improvements exist and returns $D$. 

Since the size of $D$ decreases by at least one after each iteration, the number of iterations in the algorithm is at most $n+m$. Moreover, each iteration takes $O((n+m)^k)$ time, because we need to check every subset of size at most $k$. Therefore, the total running time of the algorithm is $O((n+m)^{k+1}) = (n+m)^{O(1/\eps^2)}$.

\subsection{Approximation ratio} \label{sec:PTASratio}
We now analyze the approximation ratio of our algorithm based on the framework of~\cite{Chan2012,Mustafa2010}.
Let $D$ denote the locally optimal dominating set obtained by the algorithm, and let $O$ denote an optimal dominating set of $G$. We assume, w.l.o.g., that $D \cap O = \emptyset$, i.e., there is no interval that belongs to both $D$ and $O$. (Common intervals are ignored since they belong to both solutions and do not affect the analysis.)

\begin{lemma}[Locality Condition]\label{lemma:LocalityCondition}
There exists a planar graph $H = (D \cup O,E_H)$, such that for every interval $x \in \I$, there is an interval $u \in D$ dominating $x$ and an interval $v \in O$ dominating $x$, such that $(u,v) \in E_H$.
\end{lemma}

Section~\ref{sec:LocalityCondition} is devoted to the proof of Lemma~\ref{lemma:LocalityCondition}. Once we establish the locality condition, the analysis of the algorithm is the same as in~\cite{Mustafa2010}. For completeness, we include the proof specialized to our setting. 

Since the graph $H = (D \cup O,E_H)$ is planar, the following planar separator theorem can be used.
\begin{theorem}[Frederickson~\cite{Federickson1987}] \label{thm:Frederickson}
There are constants $c_1$, $c_2$, and $c_3$, such that for any planar graph $G=(V,E)$ with $n$ vertices and a parameter $1 \le r \le n$, there exists a set $S \subseteq V$ of size at most $\frac{c_1n}{\sqrt{r}}$ and a partition of $V\setminus S$ into $\frac{n}{r}$ sets $V_1,V_2,\dots,V_{\frac{n}{r}}$ satisfying: 
\begin{enumerate}
    \item $|V_i| \le c_2r$, for every $1 \le i \le \frac{n}{r}$, 
    \item $N(V_i) \cap V_j = \emptyset$, for every $i \neq j$, and
    \item $|N(V_i) \cap S| \le c_3\sqrt{r}$, for every $1 \le i \le \frac{n}{r}$, 
\end{enumerate}
where $N(V')=\{u \in V\setminus V' : \exists v \in V' \text{ such that } (u,v) \in E\}$.
\end{theorem}

We apply Theorem~\ref{thm:Frederickson} to the graph $H = (D \cup O,E_H)$ in Lemma~\ref{lemma:LocalityCondition}, setting $n = |D| + |O|$ and $r=\frac{k}{c_2}$, where $c_2$ is the constant in Theorem~\ref{thm:Frederickson}. 
Hence, $|V_i| \le k$, for every $1 \le i \le \frac{n}{r}$. 

Let $D_i = V_i \cap D$ and $O_i = V_i \cap O$. Thus, $|D_i| \le |O_i| + |N(V_i) \cap S|$; otherwise, replacing $D_i$ with $O_i \cup (N(V_i) \cap S)$ results in a dominating set of size smaller than $D$, which contradicts that $D$ is locally optimal. Therefore, we have
\[
\begin{aligned}
|D| & \le |S| + \sum_{i=1}^{n/r}|D_i| \  \le |S| +  \sum_{i=1}^{n/r}|O_i| + \sum_{i=1}^{n/r} |N(V_i) \cap S|.
\end{aligned}
\]
Since $O_1,O_2,\dots,O_{\frac{n}{r}}$ are disjoint, we have
\[
\begin{aligned}
|D| \le |O| + |S| + \sum_{i=1}^{n/r} |N(V_i) \cap S|. 
\end{aligned}
\]
Since $\sum\limits_{i=1}^{n/r}|N(V_i) \cap S| \le \frac{c_3 n}{\sqrt{r}} = \frac{c_3(|D|+|O|)}{\sqrt{r}}$, $|S| \le \frac{c_1n}{\sqrt{r}}= \frac{c_1(|D|+|O|)}{\sqrt{r}}$, and $r=\frac{k}{c_2}$, we have 
\[
\begin{aligned}
|D| \le |O| + \frac{c_1(|D|+|O|)}{\sqrt{k}/\sqrt{c_2}} + \frac{c_3(|D|+|O|)}{\sqrt{k}/\sqrt{c_2}}.
\end{aligned}
\]
By setting $c_4 = (c_1+c_3)\sqrt{c_2}$, we have
\[
\begin{aligned}
|D| \le |O| + \frac{c_4(|D|+|O|)}{\sqrt{k}}.
\end{aligned}
\]
Thus, we have
\[
\begin{aligned}
|D| \le \frac{1+c_4/\sqrt{k}}{1-c_4/\sqrt{k}} |O|.
\end{aligned}
\] 
By choosing $k = \frac{c}{\eps^2}$ for sufficiently large constant $c \ge (c_4(\eps+2))^2$, we have $|D| \le (1+\eps) |O|$.

\begin{theorem}
Given a bipartite circle graph $G$ on $n$ red and $m$ blue vertices, and $\eps > 0$, there exists a $(n+m)^{O(1/\eps^2)}$-time algorithm that computes a dominating set of $G$ of size $(1 +\eps)OPT$, where $OPT$ is the size of a minimum dominating set of $G$.
\end{theorem}

\subsubsection{Establishing the Locality Condition} \label{sec:LocalityCondition}
We construct a planar \textit{exchange} graph $H = (D \cup O,E_H)$ satisfying the locality condition of Lemma~\ref{lemma:LocalityCondition} as follows. 
Recall that, for each interval $I \in \I$, $s(I)$ and $e(I)$ denote the starting and ending points (i.e., the left and right endpoints) of $I$, respectively, and that $C(I)$ denotes the set of intervals in $\I$ that intersect $I$, i.e., $C(I) = \{I' \in \I:s(I') < s(I) < e(I') < e(I) \text{ or } s(I) < s(I') < e(I) < e(I')\}$. 
We assume throughout that the endpoints of the intervals of $\I$ are distinct.

Let $I$ be an interval of $\I$ and assume, w.l.o.g., that $I$ is blue.
We define the exchange edge contributed by $I$ as follows. 
\medskip
\\ \textbf{Rule~1:} If $D \cap C(I) \neq \emptyset$ and $O \cap C(I) \neq \emptyset$, then we select an interval $u \in D \cap C(I)$ and an interval $v \in O \cap C(I)$, such that $u$ and $v$ are consecutive (along $I$) among all intervals of $D \cup O$ intersecting $I$, and add $(u,v)$ to $E_H$. 
\medskip
\\ \textbf{Rule~2:} If $D \cap C(I) = \emptyset$, then $I \in D$ and there exists at least one interval $v \in O$ dominating $I$. We select the shortest interval $v \in O$ that dominates $I$ (if there is more than one such interval, we select one arbitrarily) and add $(I,v)$ to $E_H$.  
\medskip
\\ \textbf{Rule~3:} Similarly, if $O \cap C(I) = \emptyset$, then $I \in O$ and there exists at least one interval $v \in D$ dominating $I$. We select the shortest interval $u \in D$ that dominates $I$ (if there is more than one such interval, we select one arbitrarily) and add $(u,I)$ to $E_H$. 
\begin{figure}[htb]
    \centering
    \includegraphics[width=0.67\textwidth]{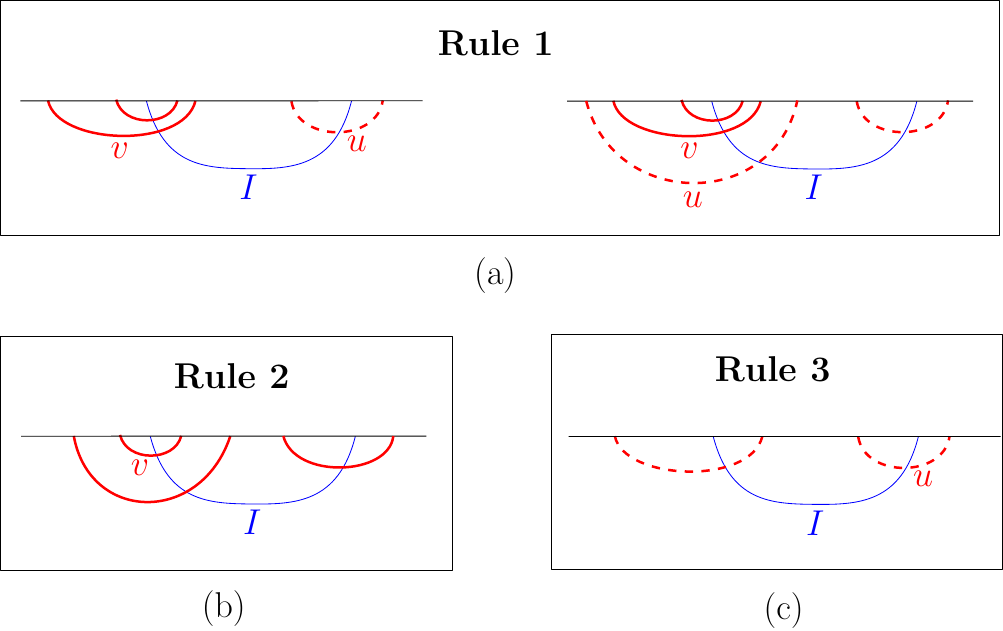}
    \caption{The exchange edge produced by a blue interval $I$. (Here and after, Solid curves indicate intervals of $O$ and dashed curves indicate intervals of $D$.) (a) An edge $(u,v)$ is added by Rule~1. $u$ and $v$ can be disjoint or nested. (b) An edge $(I,v)$ is added by Rule~2. (c) An edge $(u,I)$ is added by Rule~3.}
    \label{fig:Rules}
\end{figure}
%

Notice that the edges produced by \textbf{Rule~1} connect intervals of the same color, while the edges produced by \textbf{Rule~2} and \textbf{Rule~3} connect intervals of different colors. Moreover, since $D \cap O = \emptyset$, the two endpoints of the resulting exchange edge do not coincide as vertices. 
If the same edge is produced by several intervals, we keep only one copy. 

The exchange graph $H$ has a vertex set $D \cup O$ and an edge set $E_H$ containing the edges produced by the above rules over all intervals of $\I$.
We now prove that $H$ is planar. The proof is topological. 
We first consider the edges connecting intervals of the same color, i.e., the edges produced by \textbf{Rule~1}. 
Let $H_R$ be the subgraph of $H$ induced by the red intervals of $D \cup O$, and let $H_B$ be the subgraph of $H$ induced by the blue intervals of $D \cup O$.

\begin{lemma} \label{lemma:H_R}
The graph $H_R$ is planar.
\end{lemma}

\begin{proof}
Each edge in $H_R$ is generated by a blue interval of $\I$ and connects a red interval from $D$ and a red interval from $O$. 
We draw the red intervals as non-crossing arcs below a baseline $\ell$. This is possible because the red intervals do not cross each other: disjoint red intervals are drawn as disjoint arcs, and nested red intervals are drawn as nested arcs. We regard each red interval as a small vertex-region and draw the exchange edges as curves between these regions.
The blue intervals are drawn similarly.

Let $(u,v)$ be an exchange edge of $H_R$ contributed by a blue interval $b \in \I$ according to \textbf{Rule~1}. 
We draw the edge $(u,v)$ as a sub-curve of $b$ connecting the intersection points of $b$ with $u$ and $v$. 
Since the blue intervals do not cross, this drawing has no crossing edges.
\begin{figure}[htb]
    \centering
    \includegraphics[width=0.65\textwidth]{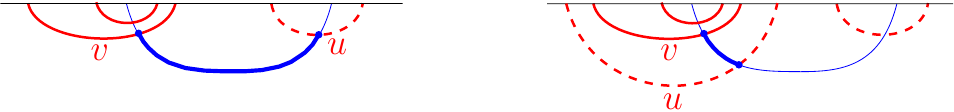}
    \caption{The drawing of an edge $(u,v)$ contributed by a blue interval $b$ according to Rule~1.}
    \label{fig:planar1}
\end{figure}

We now show that each vertex-region is safe, i.e., no exchange edge $(u,v)$ of $H_R$ can intersect the vertex-region of each red interval $r$. This is true since, for each edge $(u,v)$ that is contributed by a blue interval $b$, using \textbf{Rule~1}, $u$ and $v$ are consecutive, and hence, $r$ cannot intersect $b$ between $u$ and $v$.  Thus, the edges in the drawing do not cross any vertex-region of the red intervals.

Finally, consider two exchange edges incident to the same red interval $r$. Each such edge is attached to $r$ near the endpoint of the corresponding interval that is dominated by $r$. These attachment points appear in a well-defined left-to-right order along $r$. 
We draw the incident edges in this order around the vertex-region of $r$. 
Thus, edges that share an endpoint of $r$ meet only at $r$ and do not cross. 
Therefore, $H_R$ is planar.
\end{proof}

The following lemma is symmetric to Lemma~\ref{lemma:H_R}, and its proof is similar.

\begin{lemma} \label{lemma:H_B}
The graph $H_B$ is planar.
\end{lemma}

\begin{theorem} \label{thm:exchange-graph-planar}
The exchange graph $H$ is planar.
\end{theorem}

\begin{proof}
By Lemma~\ref{lemma:H_R}, $H_R$ is planar, and by Lemma~\ref{lemma:H_B}, $H_B$ is planar.
Moreover, the vertex sets of $H_R$ and $H_B$ are disjoint since one consists only of red intervals and the other consists only of blue intervals. 
Thus, the union of $H_R$ and $H_B$ is also planar. It remains to show that the edges produced by \textbf{Rule~2} and \textbf{Rule~3}, connecting the red intervals of $H_R$ and the blue intervals of $H_B$, can be drawn without crossings.

Assume that the drawing of $H_R$ is below the horizontal line $\ell$ and the drawing of $H_B$ is above the horizontal line $\ell'$, which is one unit above $\ell$. 
Let $b$ be a blue interval and let $(b,r)$ be the edge contributed by $b$ according to \textbf{Rule~2}. 
Hence, $b \in D$ is drawn above $\ell'$, and $r$ is the shortest red interval in $O$ that dominates $b$ is drawn below $\ell$; see Figure~\ref{fig:planar2}. 
Assume, w.l.o.g., that $r$ contains the endpoint $e(b)$ of $b$. 
Let $p(b)$ be the projection of $e(b)$ on $\ell$ and let $ref(b)$ be the reflection of $b$ below $\ell$. 
We draw the edge $(b,r)$ as a curve consisting of a vertical segment connecting $s(b)$ to $e(b)$ and the portion of $ref(b)$ between $p(b)$ and the intersection point of $ref(b)$ with $r$.

Clearly, $(b,r)$ does not cross the vertex-regions of the blue intervals of $H_B$ and does not cross the vertex-regions of the red intervals of $H_B$ that contain $r$.
Since the blue intervals do not cross each other, $(b,r)$ does not cross the edges of $H_R$, since these edges are drawn as portions of blue intervals.  
Moreover, since $D \cap C(b) = \emptyset$, no red interval in $D$ contains $p(b)$, and since $r$ is the shortest red interval in $O$ that intersects $b$, no red interval of $O$ that is contained in $r$ contains $p(b)$. 
Thus, $(b,r)$ does not cross the vertex-regions of the red intervals of $H_R$. 
\begin{figure}[htb]
    \centering
    \includegraphics[width=0.63\textwidth]{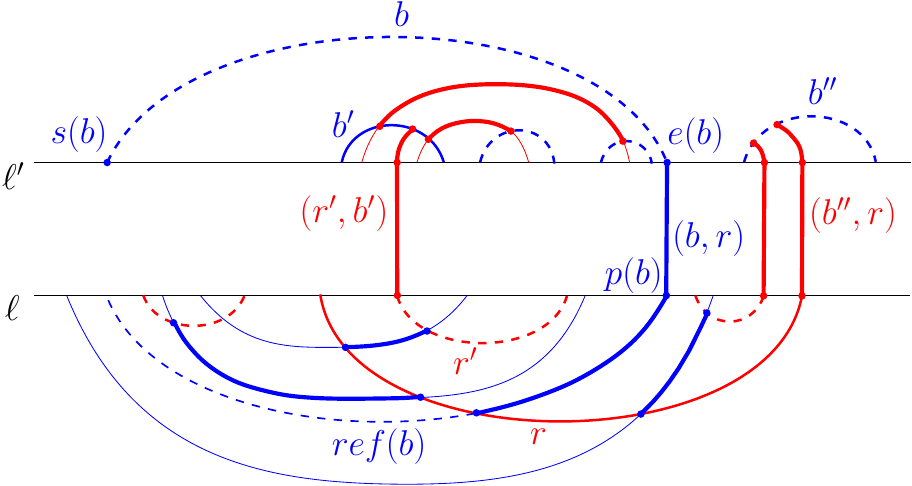}
    \caption{The drawing of the edges produced by Rule~2 and Rule~3. The edge $(b,r)$ is contributed by a blue interval $b \in D$ and a red interval $r \in O$, the edge $(r',b')$ is contributed by a red interval $r' \in D$ and a blue interval $b' \in O$, and the edge $(b'',r)$ is contributed by a blue interval $b'' \in D$ and a red interval $r \in O$.}
    \label{fig:planar2}
\end{figure}

We draw the edges contributed by the red intervals similarly (see, for example, the edge $(r',b')$ and $(b'',r)$ in Figure ~\ref{fig:planar2}). 
Hence, the edges produced by \textbf{Rule~2} and \textbf{Rule~3} can be drawn without crossings. 
Therefore, $H$ is planar.
\end{proof}

\section{Concluding Remarks}
We studied the minimum dominating set problem on bipartite circle graphs. We proved that the problem remains NP-hard even in this restricted class by presenting a reduction from Planar Monotone 3-SAT. On the positive side, we developed a polynomial-time 2-approximation algorithm based on dynamic programming and showed that a local-search algorithm yields a polynomial-time approximation scheme.

Several interesting questions remain open. In particular, it would be interesting to determine whether the problem admits an exact polynomial-time algorithm for further restricted subclasses of bipartite circle graphs or whether faster approximation schemes can be obtained.


\bibliographystyle{plainurl}

\begin{thebibliography}{10}

\bibitem{Booth1982}
K.~S. Booth and J.~H. Johnson.
\newblock Dominating sets in chordal graphs.
\newblock {\em SIAM Journal on Computing}, 11(1):191--199, 1982.
\newblock \href {https://doi.org/10.1137/0211015} {\path{doi:10.1137/0211015}}.

\bibitem{Bousquet2014}
N.~Bousquet, D.~Gonçalves, G.~B. Mertzios, Ch. Paul, I.~Sau, and S.~Thomassé.
\newblock Parameterized domination in circle graphs.
\newblock {\em Theory of Computing Systems}, 54(1):45--72, 2014.
\newblock \href {https://doi.org/10.1007/s00224-013-9478-8}
  {\path{doi:10.1007/s00224-013-9478-8}}.

\bibitem{Brandstadt1987}
A.~Brandstädt and D.~Kratsch.
\newblock On domination problems for permutation and other graphs.
\newblock {\em Theoretical Computer Science}, 54(2):181--198, 1987.
\newblock \href {https://doi.org/10.1016/0304-3975(87)90128-9}
  {\path{doi:10.1016/0304-3975(87)90128-9}}.

\bibitem{Chan2012}
T.~M. Chan and S.~Har-Peled.
\newblock Approximation algorithms for maximum independent set of pseudo-disks.
\newblock {\em Discrete \& Computational Geometry}, 48(2):373--392, 2012.
\newblock \href {https://doi.org/10.1007/s00454-012-9417-5}
  {\path{doi:10.1007/s00454-012-9417-5}}.

\bibitem{Cockayne1980}
E.~J. Cockayne, R.~M. Dawes, and S.~T. Hedetniemi.
\newblock Total domination in graphs.
\newblock {\em Networks}, 10(3):211--219, 1980.
\newblock \href {https://doi.org/10.1002/net.3230100304}
  {\path{doi:10.1002/net.3230100304}}.

\bibitem{Corcoran2021}
P.~Corcoran and A.~Gagarin.
\newblock Heuristics for k-domination models of facility location problems in
  street networks.
\newblock {\em Computers \& Operations Research}, 133:105368, 2021.
\newblock \href {https://doi.org/10.1016/j.cor.2021.105368}
  {\path{doi:10.1016/j.cor.2021.105368}}.

\bibitem{Corneil1984}
D.~G. Corneil and Y.~Perl.
\newblock Clustering and domination in perfect graphs.
\newblock {\em Discrete Applied Mathematics}, 9(1):27--39, 1984.
\newblock \href {https://doi.org/10.1016/0166-218X(84)90088-X}
  {\path{doi:10.1016/0166-218X(84)90088-X}}.

\bibitem{Czuba2025}
M.~Czuba, M.~Jia, P.~Bródka, and K.~Musial.
\newblock Applicability of the minimal dominating set for influence
  maximization in multilayer networks.
\newblock {\em Journal of Complex Networks}, 13(6):cnaf036, 2025.
\newblock \href {https://doi.org/10.1093/comnet/cnaf036}
  {\path{doi:10.1093/comnet/cnaf036}}.

\bibitem{Damaschke1989}
P.~Damaschke.
\newblock The {H}amiltonian circuit problem for circle graphs is {NP}-complete.
\newblock {\em Information Processing Letters}, 32(1):1--2, 1989.
\newblock \href {https://doi.org/10.1016/0020-0190(89)90059-8}
  {\path{doi:10.1016/0020-0190(89)90059-8}}.

\bibitem{Damian2006}
M.~Damian and S.~V. Pemmaraju.
\newblock {APX}-hardness of domination problems in circle graphs.
\newblock {\em Information Processing Letters}, 97(6):231--237, 2006.
\newblock \href {https://doi.org/10.1016/j.ipl.2005.11.007}
  {\path{doi:10.1016/j.ipl.2005.11.007}}.

\bibitem{Damian-Iordache2002}
M.~Damian-Iordache and S.~V. Pemmaraju.
\newblock A $(2+\varepsilon)$-approximation scheme for minimum domination on
  circle graphs.
\newblock {\em Journal of Algorithms}, 42(2):255--276, 2002.
\newblock \href {https://doi.org/10.1006/jagm.2001.1206}
  {\path{doi:10.1006/jagm.2001.1206}}.

\bibitem{deBerg2012}
M.~De~Berg and A.~Khosravi.
\newblock Optimal binary space partitions for segments in the plane.
\newblock {\em International Journal of Computational Geometry \&
  Applications}, 22(03):187--205, 2012.
\newblock \href {https://doi.org/10.1142/S0218195912500045}
  {\path{doi:10.1142/S0218195912500045}}.

\bibitem{Farber1982}
M.~Farber.
\newblock Independent domination in chordal graphs.
\newblock {\em Operations Research Letters}, 1(4):134--138, 1982.
\newblock \href {https://doi.org/10.1016/0167-6377(82)90015-3}
  {\path{doi:10.1016/0167-6377(82)90015-3}}.

\bibitem{Federickson1987}
G.~N. Federickson.
\newblock Fast algorithms for shortest paths in planar graphs, with
  applications.
\newblock {\em SIAM Journal on Computing}, 16(6):1004--1022, 1987.
\newblock \href {https://doi.org/10.1137/0216064} {\path{doi:10.1137/0216064}}.

\bibitem{Garey1978}
M.~R. Garey and D.~S. Johnson.
\newblock {\em Computers and Intractability: A Guide to the Theory of
  NP-Completeness}.
\newblock Freeman, San Francisco, 1978.

\bibitem{Gavril1973}
F.~Gavril.
\newblock Algorithms for a maximum clique and a maximum independent set of a
  circle graph.
\newblock {\em Networks}, 3(3):261--273, 1973.
\newblock \href {https://doi.org/10.1002/net.3230030305}
  {\path{doi:10.1002/net.3230030305}}.

\bibitem{Gavril2008}
F.~Gavril.
\newblock Minimum weight feedback vertex sets in circle graphs.
\newblock {\em Information Processing Letters}, 107(1):1--6, 2008.
\newblock \href {https://doi.org/10.1016/j.ipl.2007.12.003}
  {\path{doi:10.1016/j.ipl.2007.12.003}}.

\bibitem{Goddard2013}
W.~Goddard and M.~A. Henning.
\newblock Independent domination in graphs: A survey and recent results.
\newblock {\em Discrete Mathematics}, 313(7):839--854, 2013.
\newblock \href {https://doi.org/10.1016/j.disc.2012.11.031}
  {\path{doi:10.1016/j.disc.2012.11.031}}.

\bibitem{Haynes2017}
T.~W. Haynes, S.~Hedetniemi, and P.~Slater.
\newblock {\em Domination in Graphs: Volume 2: Advanced Topics}.
\newblock CRC Press, 2017.

\bibitem{Johnson1985}
D.~S. Johnson.
\newblock The {NP}-completeness column: an ongoing guide.
\newblock {\em Journal of Algorithms}, 6(3):434--451, 1985.
\newblock \href {https://doi.org/10.1016/0196-6774(85)90012-4}
  {\path{doi:10.1016/0196-6774(85)90012-4}}.

\bibitem{Keil1993}
J.~M. Keil.
\newblock The complexity of domination problems in circle graphs.
\newblock {\em Discrete Applied Mathematics}, 42(1):51--63, 1993.
\newblock \href {https://doi.org/10.1016/0166-218X(93)90178-Q}
  {\path{doi:10.1016/0166-218X(93)90178-Q}}.

\bibitem{Keil2006}
J.~M. Keil and L.~Stewart.
\newblock Approximating the minimum clique cover and other hard problems in
  subtree filament graphs.
\newblock {\em Discrete Applied Mathematics}, 154(14):1983--1995, 2006.
\newblock \href {https://doi.org/10.1016/j.dam.2006.03.003}
  {\path{doi:10.1016/j.dam.2006.03.003}}.

\bibitem{Kelleher1988}
L.~L. Kelleher and M.~B. Cozzens.
\newblock Dominating sets in social network graphs.
\newblock {\em Mathematical Social Sciences}, 16(3):267--279, 1988.
\newblock \href {https://doi.org/10.1016/0165-4896(88)90041-8}
  {\path{doi:10.1016/0165-4896(88)90041-8}}.

\bibitem{Laskar1983}
R.~Laskar and J.~Pfaff.
\newblock Domination and irredundance in split graphs.
\newblock In {\em Technical Report 430}. Clemson University Clemson, SC, 1983.

\bibitem{Lasker1984}
R.~Laskar, J.~Pfaff, S.~M. Hedetniemi, and S.~T. Hedetniemi.
\newblock On the algorithmic complexity of total domination.
\newblock {\em SIAM Journal on Algebraic Discrete Methods}, 5(3):420--425,
  1984.
\newblock \href {https://doi.org/10.1137/0605040} {\path{doi:10.1137/0605040}}.

\bibitem{Mustafa2010}
N.~H. Mustafa and S.~Ray.
\newblock Improved results on geometric hitting set problems.
\newblock {\em Discrete \& Computational Geometry}, 44(4):883--895, 2010.
\newblock \href {https://doi.org/10.1007/s00454-010-9285-9}
  {\path{doi:10.1007/s00454-010-9285-9}}.

\bibitem{Reddy2026}
S.~B. Reddy, A.~K. Das, A.~S. Kare, and I.~V. Reddy.
\newblock On the complexity of global {R}oman domination problem in graphs,
  2026.
\newblock \href {https://arxiv.org/abs/2601.09167} {\path{arXiv:2601.09167}}.

\bibitem{Unger1988}
W.~Unger.
\newblock On the $k$-colouring of circle-graphs.
\newblock In {\em Proceedings of the 5th Annual Symposium on Theoretical
  Aspects of Computer Science (STACS)}, pages 61--72, 1988.
\newblock \href {https://doi.org/10.1007/BFb0035832}
  {\path{doi:10.1007/BFb0035832}}.

\bibitem{Wu1999}
J.~Wu and H.~Li.
\newblock On calculating connected dominating set for efficient routing in ad
  hoc wireless networks.
\newblock In {\em Proceedings of the 3rd International Workshop on Discrete
  Algorithms and Methods for Mobile Computing and Communications (DIALM)}, page
  7–14, 1999.
\newblock \href {https://doi.org/10.1145/313239.313261}
  {\path{doi:10.1145/313239.313261}}.

\end{thebibliography}

\end{document}